\documentclass{cccg26}
\usepackage{graphicx,amssymb,amsmath}

\usepackage[T1]{fontenc}
\usepackage[utf8]{inputenc}

\usepackage{amsfonts,amsmath,amssymb}
\usepackage{cite,microtype,graphicx,hyperref}

\usepackage{zref-clever}
\zcsetup{cap}
\let\cref\zcref
\AddToHook{env/theorem/begin}{\zcsetup{countertype={theorem=theorem}}}
\AddToHook{env/lemma/begin}{\zcsetup{countertype={theorem=lemma}}}

\newtheorem{observation}[theorem]{Observation}
\zcRefTypeSetup{observation}{Name-sg=Observation,Name-pl=Observations}
\AddToHook{env/observation/begin}{\zcsetup{countertype={theorem=observation}}}

\title{Tangent Spheres and Integer Distances}

\author{David Eppstein\thanks{Department of Computer Science, University of California, Irvine. Donald Bren Hall, Irvine, CA 92697, USA. eppstein@uci.edu. Research supported in part by NSF grant CCF-2212129.}}

\index{Eppstein, David}

\begin{document}
\thispagestyle{empty}
\maketitle

\begin{abstract}
The Erd\H{o}s--Anning theorem states that any point set for which all distances are integers, in a Euclidean space of any dimension, must be either finite or collinear. We prove the same result in hyperbolic space of any dimension. A quantitative form of our result also extends for the first time to Euclidean spaces of dimension greater than two: if a set of points with integer distances in $\mathbb{E}^D$ or $\mathbb{H}^D$ has a subset of $D+1$ points in general position whose diameter is $d$, then the whole set has size $O(D(d+1)^D)$. To prove these results we formulate a lemma that, if the graph of external tangencies of a system of spheres in Euclidean or hyperbolic space contains a $K_{a,b}$ subgraph for $a,b\ge 3$, then the sets of spheres on each side of this biclique have centers that lie on a hyperplane. This lemma also implies that, in multilateration (determining a position from differences of distances to known landmarks), $D+1$ non-coplanar landmarks always suffice to limit the position to two possibilities.
\end{abstract}

\section{Introduction}

In a 1944 paper, Norman Anning and Paul Erd\H{o}s proved that every point set in the Euclidean plane with all distances integers is either finite or collinear. Their proof uses trigonometric inequalities specific to Euclidean geometry and does not provide an explicit bound on the size of a non-collinear point set (which can have any finite cardinality); they remark without detail that similar techniques extend to all higher-dimensional Euclidean spaces~\cite{AnnErd-BAMS-45}. Soon after, Erd\H{o}s published an alternative, conceptually simpler, and more quantitative proof, which he boiled down to four lines of text: any two points $p$ and $q$ at integer distance $d$ determine a system of $d+1$ hyperbolae containing all points at integer distances from $p$ and $q$. Another point $r$, not on line $pq$, determines another such system of hyperbolae, and these two systems intersect in at most four points per pair of hyperbolae. Thus, whenever an integer-distance set contains a non-collinear triple of points of of diameter $d$ it contains at most $4(d+1)^2$ points~\cite{Erd-BAMS-45}. Again, Erd\H{o}s wrote that this extends to higher dimensions, but without detail, and his simplified proof hides some case analysis (some of the hyperbolae degenerate to lines or rays which must be verified to be distinct) raising doubts about whether the possible degeneracies among multiple hyperbolic surfaces might get messier in higher dimensions.

Subsequent work has provided additional quantitative bounds: if an integer-distance point set in the plane has diameter $d$, it has $O(d)$ points~\cite{Sol-DCG-03} (obviously tight for collinear sets) and if, in addition, it is not collinear, it has $d^{O(1/\log\log d)}$ points~\cite{GreIliPel-24}.

The use of trigonometry in the Anning--Erd\H{o}s proof, and of algebraic geometry  in the Erd\H{o}s proof (in the intersection numbers of degree-two algebraic curves), limit them to Euclidean geometry. Another recent proof by the author~\cite{Epp-JoCG-26} instead uses the topological properties of additively-weighted Voronoi diagrams. In this way it extends versions of the same result, that an integer-distance set is either collinear (belonging to a geodesic) or finite, to many other two-dimensional metrics including all strictly convex distance functions, all surfaces of three-dimensional convex bodies, and all complete two-dimensional Riemannian metrics of bounded genus, including the hyperbolic plane. However, there are obstacles to extending this Voronoi diagram based proof to higher dimensions, especially for convex distance functions~\cite{IckKleLe-FI-95}.

In our earlier work~\cite{Epp-JoCG-26} we asked whether the Erd\H{o}s--Anning theorem can be proven for higher-dimensional hyperbolic spaces, one of the most well-behaved of non-Euclidean geometries. The lack of detail about high dimensions in the original papers of Erd\H{o}s and Anning also raises the question of how the quantitative bounds known for the Euclidean plane extend to higher dimensional Euclidean spaces. Here we answer both of these questions. We prove a bound of $2(D+1)(d+1)^D+D+1$ on the size of an integer-distance point set that contains a general-position subset of diameter $d$, both in in $D$-dimensional Euclidean space $\mathbb{E}^D$ or hyperbolic space $\mathbb{H}^D$. This new bound has the same form (with a different constant factor) as Erd\H{o}s's bound for $D=2$.

Rather than using trigonometry, hyperboloids, or Voronoi diagrams, our proof is based on the tangency and orthogonality properties of spheres. In this way, it connects the problem of integer distances to the theories of circle packing, sphere packing, and inversive geometry, and to the lifting transformation of computational geometry. In particular, we prove a key lemma that may be of independent interest: if the graph of external tangencies of a system of spheres in Euclidean or hyperbolic space of any dimension contains a $K_{a,b}$ subgraph for $a,b\ge 3$, then the centers of the spheres on each side of this biclique must lie on a hyperplane. 

\subsection{Applications to Multilateration}

\begin{figure}[t]
\centering\includegraphics[width=0.4\textwidth,alt={Spheres on an ellipse, and the Dupin cyclide that they approximate}]{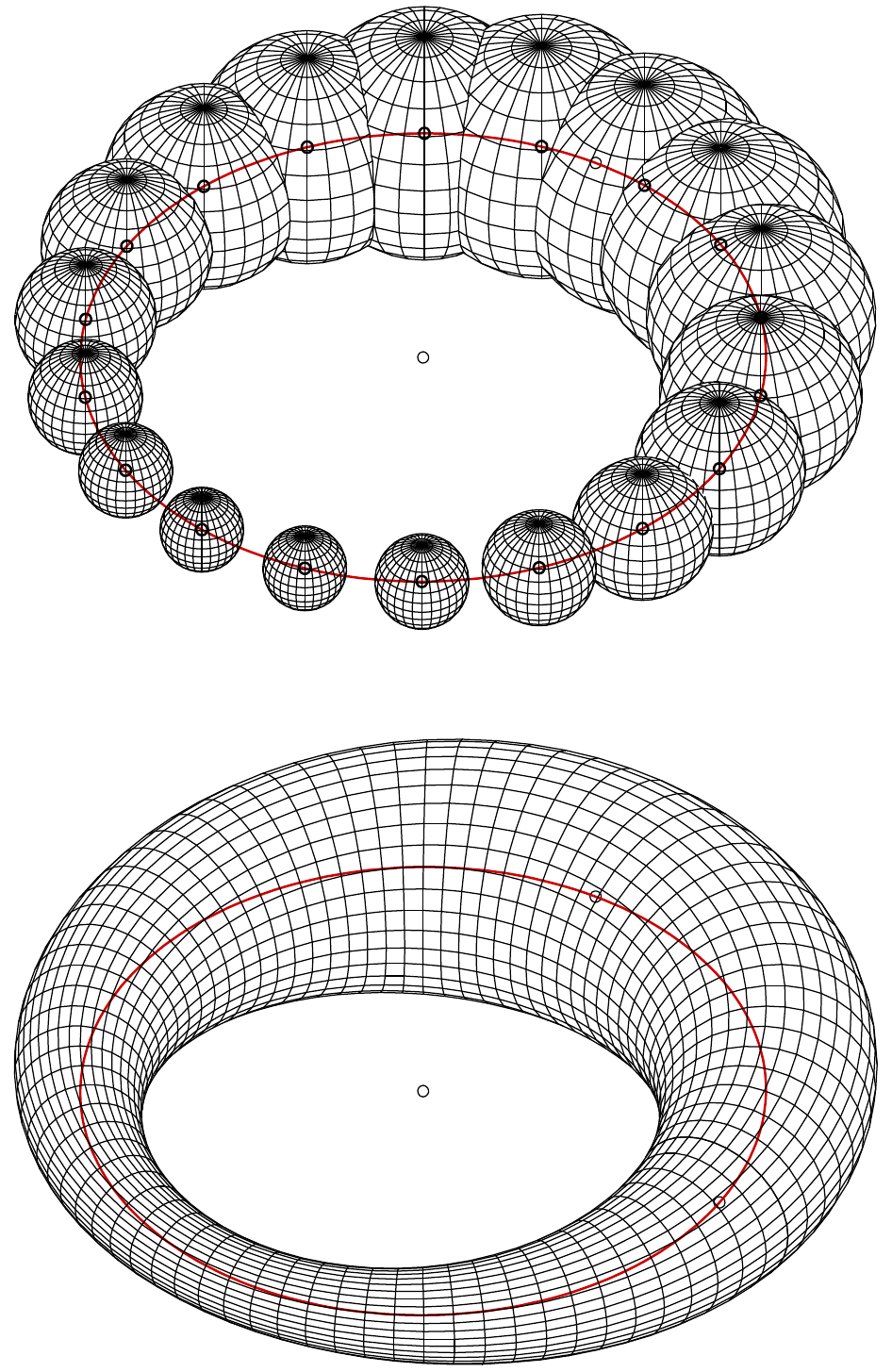}
\caption{A Dupin cyclide as an envelope of spheres centered on an ellipse. CC-BY-SA image by Ag2gaeh from \url{https://commons.wikimedia.org/wiki/File:Zyklide-kanalfl-1-2.svg}.}
\label{fig:cyclide}
\end{figure}

Our lemma has practical implications for $D$-dimensional \emph{multilateration}, methods for determining a position from the differences of distances to a set of known landmarks. This arises in geolocation using the GPS system, where the use of differences corrects for the unknown inaccuracy of the clock in a GPS device, relative to coordinated universal time~\cite{AbeCha-TAES-91}. When determining a position from distances to landmarks, $D$ landmarks (in general position) suffice to limit the position to two possibilities, but when using differences of distances, one more landmark is needed.
If the closest of the landmarks is selected as a reference, and one constructs spheres around each landmark whose radius is the difference of distances with this reference landmark (plus some small $\varepsilon$ to make the smallest sphere non-degenerate), then the position to be determined is the center of another sphere externally tangent to all of these spheres. Typically there are only two choices for this position, only one of which is near the surface of the Earth. Alternatively one could disambiguate these two positions using another landmark.

Our lemma allows us to make a more precise statement about  exactly which sets of landmarks lead to degeneracies in this location process.
It implies that the only degenerate cases for this multilateration problem arise when the landmarks lie on a common hyperplane: $D+1$ landmarks that are not coplanar always suffice to limit the determined position to two points. For, if there could be three positions each consistent with the known differences of distances to a system of landmarks, then the spheres centered on the landmarks, and the tangent spheres to them centered on the positions, would form two families of tangent spheres to which our lemma applies, showing that the landmarks must be coplanar. 

A classical construction shows that for $D=3$, coplanar landmarks are indeed a degenerate case. Even $D+2$ coplanar landmarks are insufficient to localize the position to a finite set of possibilities. The construction uses inversion on a torus.  Any number of congruent spheres centered on a circle are internally tangent to a torus or self-crossing torus, and this torus and the spheres within it are externally tangent to arbitrarily many spheres centered on a line perpendicular to the circle of sphere centers. Inversion through a sphere preserves the tangencies, taking the torus to a \emph{Dupin cyclide} (\cref{fig:cyclide}), a surface shaped like a deformed torus that again has two infinite families of spheres, each tangent to one side of the surface and through it tangent to all of the spheres in the other family. However, inversion does not map the circle of sphere centers to another circle; instead, each Dupin cyclide has two perpendicular conic curves as its \emph{directrices}, the curves formed by the centers of its tangent spheres~\cite{HilCoh-52,ChaDutHog-VC-89}. If five landmarks in $\mathbb{R}^3$ are coplanar but otherwise in general position, then they lie on a conic, which forms the directrix of a Dupin cyclide. Five spheres, centered at the landmarks and tangent to the cyclide, are tangent to infinitely many spheres centered on the other directrix, whose centers cannot be distinguished from each other by multilateration.

\section{Integer Distances}

In this section we state but do not prove our lemma on tangent spheres, and use it to bound the number of points in an integer-distance subset of any Euclidean or hyperbolic space.

By a \emph{sphere} in Euclidean space $\mathbb{E}^D$ or hyperbolic space $\mathbb{H}^D$, we always mean a $(D-1)$-sphere, a full-dimensional hypersphere. Define an $(a,b)$-\emph{tangency biclique} to consist of two systems of $a$ spheres $A_i$ and $b$ spheres $B_j$, in a Euclidean or hyperbolic space, such that each $A_i$ and $B_j$ are externally tangent. Within $A_i$ and $B_j$ the spheres are not required to be disjoint, but they must be distinct. For example \emph{Soddy's hexlet} can form a $(3,6)$-tangency biclique in $\mathbb{E}^3$, although many images of Soddy's hexlet, including its earliest known appearance on an 1822 Japanese \emph{sangaku} tablet, instead depict one sphere surrounding eight others. $\mathbb{E}^3$ contains $(a,b)$-tangency bicliques for arbitrary $a$ and $b$, consisting of $a$ spheres each internally tangent to a torus on a longitudinal circle and $b$ spheres each externally tangent on a latitudinal circle. The main technical lemma of our proof is:

\begin{lemma}
\label{lem:biclique}
In any $(3,b)$-tangency biclique in $\mathbb{E}^D$ or $\mathbb{H}^D$ defined by spheres $A_i$ and $B_j$, the centers of the spheres $B_j$ all lie on a common hyperplane.
\end{lemma}

From this we can prove that when a point set has integer distances to a small subset of points in general position, the whole set has bounded size:

\begin{theorem}
\label{thm:basis}
Let $C$ be a set of $D+1$ points in $\mathbb{E}^D$ or $\mathbb{H}^D$, not all on a hyperplane. Then the number of points distinct from $C$ that have integer distances to all points in $C$ is at most $2(D+1)\,(\operatorname{diam}(C)+1)^D$.
\end{theorem}

\begin{proof}
Number the points in $C$ as $c_0,\dots c_D$. If any point $p$ has integer distances to all points in $C$, then by the triangle inequality each two of these distances differ by at most $\operatorname{diam}(C)$. Let $d_{\min}=\min\operatorname{dist}(p,c_i)$. Then we can define a vector of distance differences, $\delta^{(p,C)}$, by $\delta^{(p,C)}_i=\operatorname{dist}(p,c_i)-d_{\min}$. There are at most $(D+1)(\operatorname{diam}(C)+1)^D$ choices for this vector, because all of its coefficients are integers in the range $[0,\operatorname{diam}(C)]$ and at least one of them is zero. We claim that each choice can be the vector of distance differences for at most two choices of the point~$p$.

To prove this, consider any vector $\delta\in [0,\operatorname{diam}(C)]^{D+1}$,
and from it construct a system of spheres $B_i$ centered at each point $c_i$ with radius $\delta_i+\tfrac12$. (Here, the addition of $\tfrac12$ ensures that these spheres have nonzero radius, avoiding degenerate cases, while remaining small enough that none of them contain $p$.) Then each point $p_i$ that has $\delta$ as its vector of distance differences is the center of a sphere $A_i$ with radius $d_{\min}-\tfrac12$ that is externally tangent to each sphere $B_i$. Thus, if there could exist three such points $p_i$, the spheres $A_i$ and $B_i$ would form a $(3,D+1)$-tangency biclique with the centers of its spheres $B_i$ not all on a hyperplane, contradicting \cref{lem:biclique}.
\end{proof}

The finiteness of integer-distance point sets in Euclidean and hyperbolic spaces follows:

\begin{theorem}
Let $S$ be any set of points in $\mathbb{E}^D$ or $\mathbb{H}^D$ among which all distances are integers. Then $S$ is finite.
\end{theorem}

\begin{proof}
Choose a subset $C$ of $S$ in general position whose affine span equals the affine span of $S$. By \cref{thm:basis}, applied to the affine span,  $S$ has at most $2|C|(\operatorname{diam}(C)+1)^{|C|-1}$ additional points.
\end{proof}

Unlike Euclidean geometry, hyperbolic geometry has an intrinsic length scale, but our proof does not make use of this: it applies to any scaling of hyperbolic distance.

\section{Proof of the Key Lemma}

In this section we prove \cref{lem:biclique}, breaking the proof into several smaller lemmas.

\subsection{Lifting Transformation}

In order to work with spheres in $\mathbb{E}^D$ and $\mathbb{H}^D$, it is convenient to use the lifting transformation familiar in computational geometry, an equivalence between spheres in $\mathbb{E}^D$ and hyperplanes in $\mathbb{E}^{D+1}$. Instead of using the now-common version of this transformation in which we project $\mathbb{E}^D$ vertically onto a hyperplane in $\mathbb{E}^{D+1}$, we use the original version, based on an embedding of $\mathbb{E}^D$ into $\mathbb{E}^{D+1}$ and stereographic projection from $\mathbb{E}^D$ to a tangent sphere $\mathbb{S}^D$ in $\mathbb{E}^{D+1}$~\cite{Bro-IPL-79}. This projection maps spheres (and hyperplanes, viewed as degenerate spheres) in $\mathbb{E}^D$ to spheres in $\mathbb{S}^D$, and can be chosen to map the ball of a Poincar\'e model of $\mathbb{H}^D$ to a hemisphere on $\mathbb{S}^D$. For a non-degenerate sphere in $\mathbb{E}^D$ or $\mathbb{H}^D$, the exterior of the sphere maps to the side of its image in $\mathbb{S}^D$ that contains the pole of the stereographic projection, and the interior of the sphere maps to the side that does not contain the pole, so we can define interior and exterior in a consistent way on $\mathbb{S}^D$. As well as modeling hyperbolic spheres, the spheres on $\mathbb{S}^D$ also model the hyperbolic sphere at infinity, hyperbolic hyperplanes, hyperspheres, or horospheres, with the last three types of objects distinguished by being tangent to the sphere at infinity, crossing it at a non-right angles, or crossing it at a right angle respectively.

Two spheres in any of these $D$-dimensional spaces are \emph{orthogonal} if there exists a line through a point of their intersection that is tangent to one of them and perpendicular to the other. Such a line can be obtained as the radius of either sphere through any point of intersection. This property is preserved by the lifting transformation, under which the point of intersection is lifted to a point on  $\mathbb{S}^D$ through which there again exists a line tangent to one sphere (and to ${S}^D$) and perpendicular to the other.

Spheres on $\mathbb{S}^D$ correspond to hyperplanes and points in $\mathbb{E}^{D+1}$ in two dual ways. A hyperplane that intersects $\mathbb{S}^D$ does so in a sphere, and every sphere on $\mathbb{S}^D$ is the intersection of $\mathbb{S}^D$ with a hyperplane. From any point $p$ exterior to $\mathbb{S}^D$, the lines of tangency to $\mathbb{S}^D$ meet $\mathbb{S}^D$ in a sphere (the horizon as viewed from $p$), and every sphere on $\mathbb{S}^D$ is the horizon of a point $p$ (possibly at infinity, for a great sphere of $\mathbb{S}^D$). When finite, the viewpoint defining a sphere can be obtained from the hyperplane defining the sphere, by inverting through $\mathbb{S}^D$ the closest point on the hyperplane to the center of $\mathbb{S}^D$.

Two spheres on $\mathbb{S}^D$ are orthogonal if and only if the viewpoint defining one sphere as its horizon belongs to the hyperplane defining the other sphere as its intersection. If so, any line through the viewpoint and a point of intersection of the spheres is tangent to the intersection sphere and perpendicular to the horizon sphere. Every line tangent to an intersection sphere lies in the defining hyperplane, and every line perpendicular to a horizon sphere and tangent to $\mathbb{S}^D$ passes through the viewpoint, so every pair of orthogonal spheres must have this form.

\subsection{Central Hyperplanes}

\begin{figure}[t]
\centering\includegraphics[width=0.8\columnwidth,alt={A blue circle crossed at right angles by nine black circles and three red lines}]{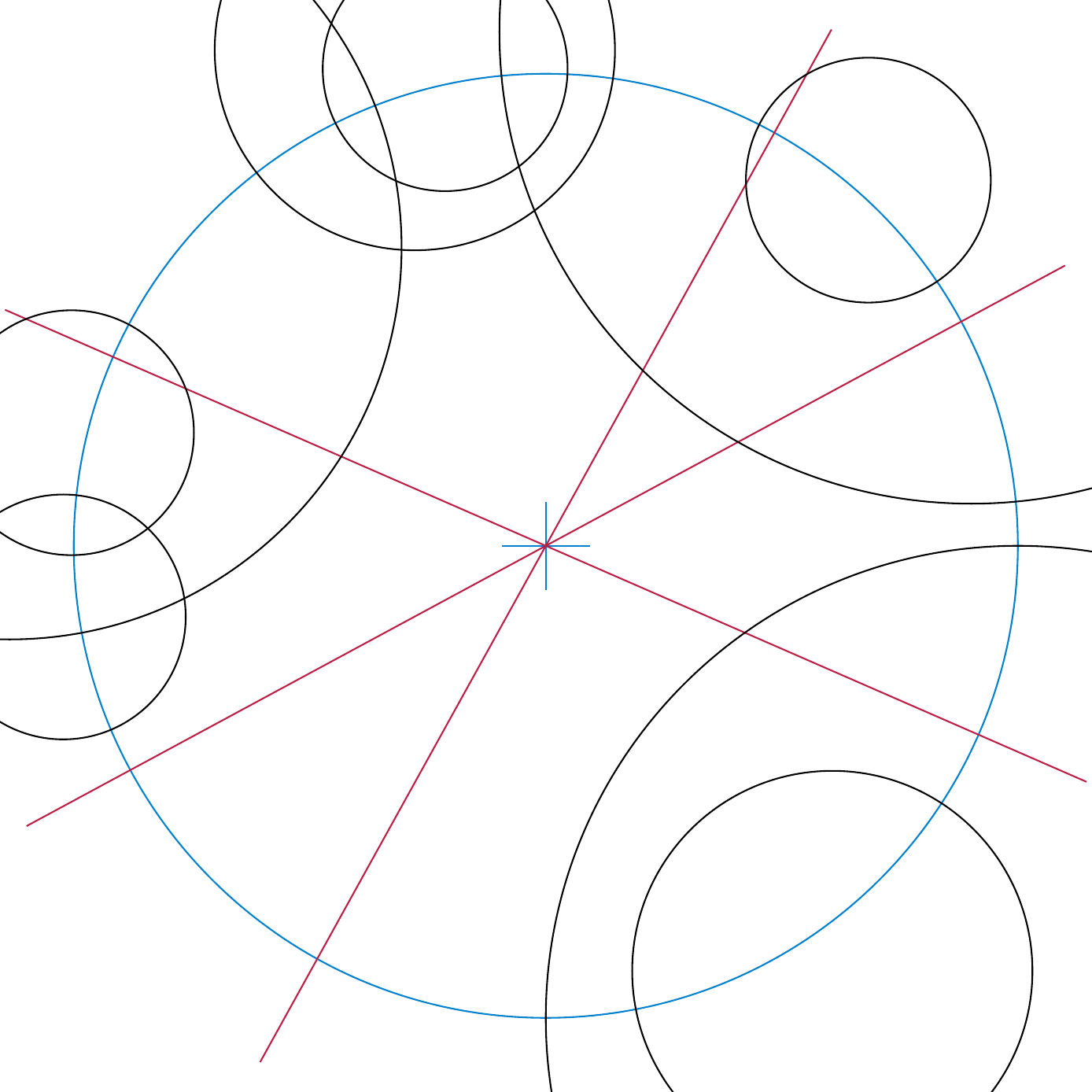}
\caption{If spheres or hyperplanes cross a given sphere (blue) at right angles, then only the hyperplanes (red) pass through the center of the given sphere (\cref{lem:central-hyperplane}).}
\label{fig:orthocircles}
\end{figure}

It is possible to characterize Euclidean or hyperbolic hyperplanes using the lifting transformation. Hyperplanes in $\mathbb{E}^D$ correspond to intersections of $\mathbb{S}^D$ with hyperplanes through the pole of stereographic projection and to horizons from viewpoints on the hyperplane tangent to the pole. Hyperplanes in $\mathbb{H}^D$ correspond to intersections by hyperplanes through the point at vertical infinity and to horizons from viewpoints on the hyperplane through the equator of $\mathbb{S}^D$. However, instead of using these lifting-based characterizations, it is more convenient for us to use a characterization in the unlifted Euclidean or hyperbolic geometry.
 
\begin{lemma}
\label{lem:central-hyperplane}
In Euclidean or hyperbolic space, if $S_0$ is a bounded sphere and $S_1$ is a sphere, hyperplane, hypersphere, or horosphere orthogonal to $S_0$, then $S_1$ is a Euclidean or hyperbolic hyperplane (respectively) if and only if $S_1$ passes through the center of $S_0$ (\cref{fig:orthocircles}).
\end{lemma}

\begin{proof}
In the Euclidean case, if $S_1$ is a hyperplane orthogonal to $S_0$, it contains a line orthogonal to $S_0$, which contains a radius of $S_0$ and therefore the center of $S_0$. For any non-degenerate sphere $S_1$ orthogonal to $S_0$ , there exists a tangent hyperplane $S_2$ through any point of intersection of $S_0$ and $S_1$; $S_2$ passes through the center of $S_0$, and $S_1$ intersects $S_2$ only at this point of tangency, so $S_2$ does not pass through the center of $S_0$.

In the hyperbolic case, choose a Poincar\'e model of hyperbolic geometry in which $S_0$ is concentric with the sphere at infinity $S_\infty$. Hyperbolic hyperplanes are modeled in Euclidean space by spheres perpendicular to $S_{\infty}$; by the Euclidean case of the lemma, the hyperbolic hyperplanes through the shared center of $S_0$ and $S_{\infty}$ are modeled by Euclidean hyperplanes through this point. Because $S_0$ and $S_{\infty}$ are concentric,  each hyperbolic hyperplane through this center point is perpendicular to $S_0$. For any other (hyper,horo)sphere $S_1$ perpendicular to $S_0$, there is a tangent (Euclidean) hyperplane $S_2$ through any point of intersection with $S_0$ and $S_1$; $S_2$ passes through the center of $S_0$, and $S_1$ intersects $S_0$ only at the point of tangency, so $S_1$ does not contain the center.
\end{proof}

\subsection{Pencils of Spheres}

\begin{figure}[t]
\centering\includegraphics[width=0.8\columnwidth,alt={Two orthogonal pencils of circles}]{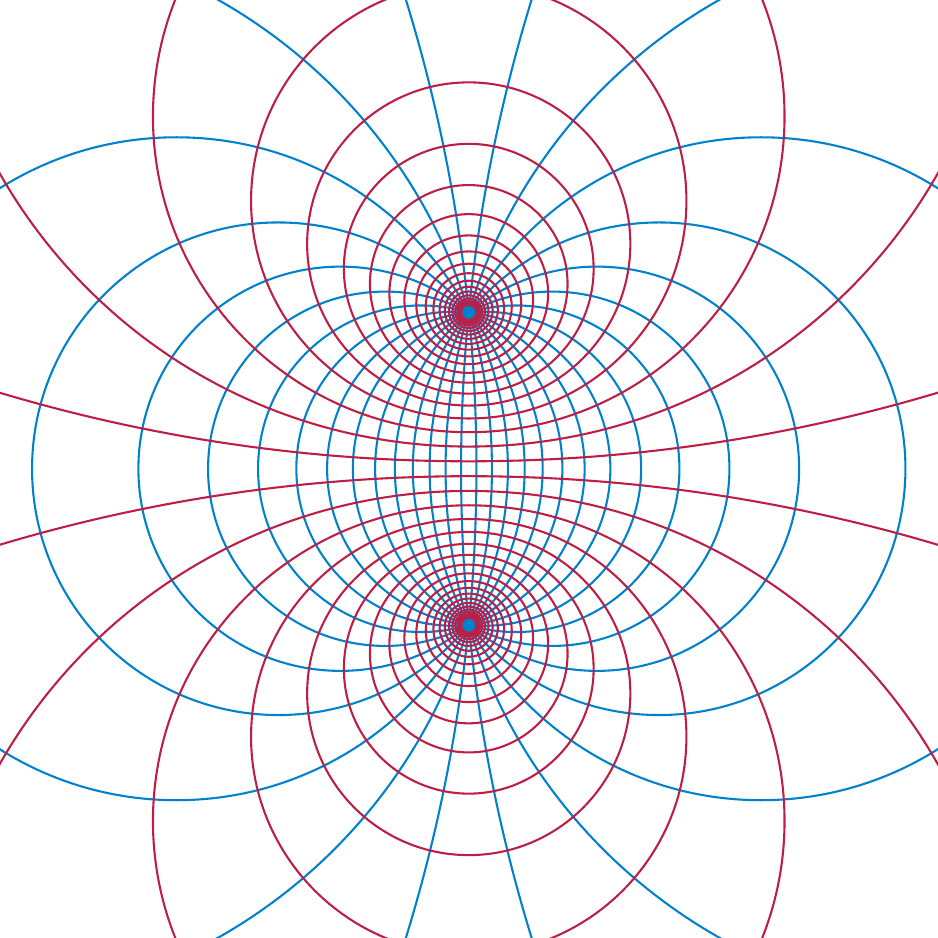}
\caption{Two orthogonal pencils of circles, from \cite{Epp-JGAA-21}. In spaces of dimension $\ge 3$, the spheres that are orthogonal to a pencil do not form another pencil.}
\label{fig:apollo}
\end{figure}

Any two hyperplanes in $\mathbb{E}^{D+1}$ intersect in a $(D-1)$-dimensional affine subspace (possibly at infinity). Define a \emph{pencil of spheres} to be the family of spheres whose defining hyperplanes, under the lifting transformation, all contain some given $(D-1)$-dimensional  affine subspace. (These can also be defined as the families of spheres determined as the horizonts of viewpoints on a fixed line, but we will not use this.) Every two spheres on $\mathbb{S}^D$ belong to a unique pencil, the pencil whose affine subspace is the intersection of the defining hyperplanes of the two spheres. For examples of pencils, see \cref{fig:apollo}.

\begin{observation}
\label{obs:pencil-orthogonality}
Let $S_0$, $S_1$, and $S_2$ be spheres on $\mathbb{S}^D$. Then $S_0$ is orthogonal to both $S_1$ and $S_2$ if and only if $S_0$ is orthogonal to every circle in the pencil containing $S_1$ and $S_2$.
\end{observation}

\begin{proof}
The orthogonality of $S_0$ to both $S_1$ and $S_2$ is equivalent, under the lifting characterization of orthogonality, to the assertion that the viewpoint from which $S_0$ is the horizon lies on the defining hyperplanes of both $S_1$ and $S_2$. This, in turn, is equivalent to the assertion that this viewpoint lies on all hyperplanes that contain the defining hyperplanes of both $S_1$ and $S_2$. But these are exactly the defining hyperplanes of the pencil containing $S_1$ and $S_2$, and applying again the lifting characterization of orthogonality, the assertion that the viewpoint for $S_0$ lies on these hyperplanes is equivalent to the assertion that it is orthogonal to all spheres in the pencil.
\end{proof}

The pencil of concentric spheres around a point does not include a (finite) hyperplane. Nevertheless, we have:

\begin{lemma}
\label{lem:orthogonal-hyperplane}
Every pencil of spheres that is orthogonal to a bounded Euclidean or hyperbolic sphere includes a Euclidean or hyperbolic hyperplane, respectively.
\end{lemma}

\begin{proof}
Let $H$ be a hyperplane in $\mathbb{E}^{D+1}$ that contains both the center point of the (lifted) bounded sphere and the defining affine subspace of the pencil. Then the sphere obtained by intersecting $H$ with $S^D$ belongs to the pencil and by \cref{lem:central-hyperplane} is a hyperplane.
\end{proof}

\begin{lemma}
\label{lem:2sphere-orthogonal}
Let $S_1$ and $S_2$ be any two spheres, degenerate spheres, hyperspheres, or horospheres in Euclidean or hyperbolic space. Then this space contains a hyperplane through the center points of all bounded spheres that are orthogonal to both $S_1$ and $S_2$.
\end{lemma}

\begin{proof}
Let $S_0$ be a bounded sphere orthogonal to $S_1$ and $S_2$; if $S_0$ does not exist, there is nothing to prove. Otherwise, by \cref{obs:pencil-orthogonality}, $S_0$ is orthogonal to the pencil of spheres containing $S_1$ and $S_2$, by \cref{lem:orthogonal-hyperplane}, this pencil contains a hyperplane $H$, and by \cref{obs:pencil-orthogonality} again, $H$ is orthogonal to all bounded spheres orthogonal to $S_1$ and $S_2$. $H$ contains all the centers of these orthogonal spheres, by \cref{lem:central-hyperplane}.
\end{proof}

\subsection{Antisimilitude}

For any two distinct spheres $A$ and $B$ in $D$-dimensional Euclidean space, there exist one or two spheres or hyperplanes (degenerate spheres) $M$  such that a sphere inversion through $M$ (or reflection across the hyperplane) takes $A$ to $B$; these have sometimes been called \emph{spheres of antisimilitude}~\cite{Cou-AMM-60}. When the interiors of $A$ and $B$ intersect, inversion through one such sphere $M$ takes the interior of $A$ to the exterior of $B$ and vice versa. When the interiors of $A$ and $B$ have a nonempty symmetric difference (the case of interest to us),  $M$ can be chosen so that inversion through $M$ preserves the interiors and exteriors of $A$ and $B$. One way to understand this~\cite{Epp-DCG-14} is to view $\mathbb{E}^D$ as the boundary of a halfspace model of $\mathbb{H}^{D+1}$, and each  sphere in $\mathbb{E}^D$ as the set of points at infinity of a hyperplane in $\mathbb{H}^{D+1}$. Then inversion through a sphere in $\mathbb{E}^D$ corresponds to reflection through the corresponding hyperplane in $\mathbb{H}^{D+1}$, and the spheres of antisimilitude of $A$ and $B$ correspond to the cell boundaries of a hyperbolic Voronoi diagram of the two hyperplanes corresponding to $A$ and $B$ in $\mathbb{H}^{D+1}$.

\begin{figure}[t]
\includegraphics[width=\columnwidth,alt={Refer to caption}]{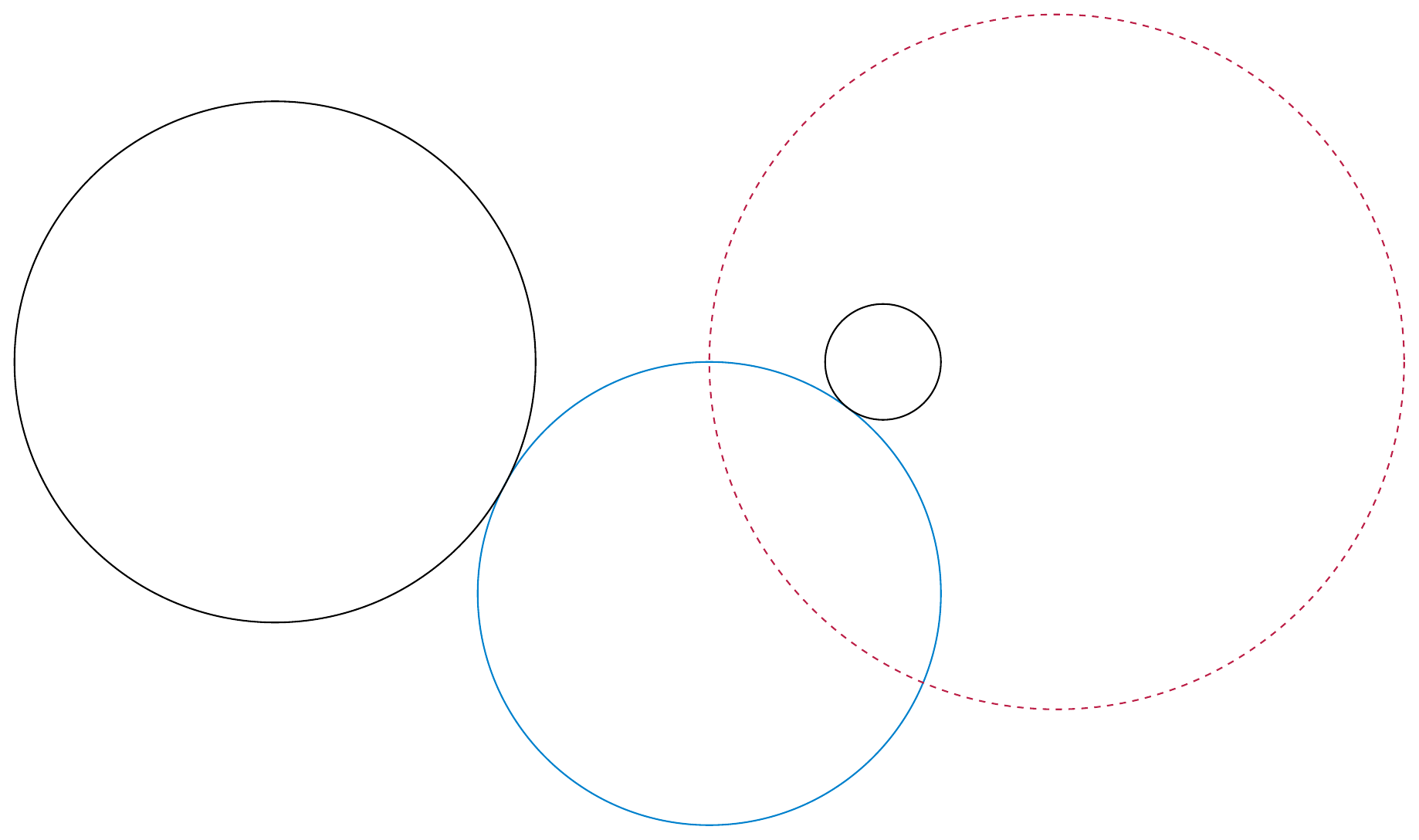}
\caption{Two given circles $A$ and $B$ (black) with their interior-preserving circle of antisimilitude $M$ (dashed red). The circle $C$ (blue) tangent to both $A$ and $B$ is orthogonal to $M$ (\cref{lem:antisimilitude-orthogonal}).}
\label{fig:antisimilitude}
\end{figure}

\begin{lemma}
\label{lem:antisimilitude-orthogonal}
If $A$ and $B$ are each externally tangent to a third sphere $C$ at distinct points, let $M$ be the sphere of antisimilitude such that inversion through $M$ takes $A$ to $B$ and maps interiors to interiors and exteriors to exteriors. Then $M$ is orthogonal to $C$.
\end{lemma}

\begin{proof}
Because the two points of tangency of $A$ and $B$ with $C$ are distinct, $A$ and $B$ each have interior points not contained in the other sphere (in small neighborhoods of the points of tangency) and $M$ exists. Because inversion through $M$ maps interior points of $A$ to interior points of $B$ and vice versa, the two points of tangency must be on opposite sides of $M$ and $C$ crosses $M$. Inversion through $M$ fixes all points of $M$, including the points of intersection with $C$, preserves the interiors and exteriors of $A$ and $B$, and maps tangencies to tangencies, so the image of $C$ after the inversion is a sphere with the same intersection with $M$ that is also externally tangent to $A$ and $B$. This could only be $C$ itself, for any other sphere with the same intersection would be contained by $C$ on one side of $M$ and would contain $C$ on the other side, causing it to pull away from one of $A$ and $B$ on the side contained by $C$ and to cross into the other of $A$ and $B$ on the other side. Because the image of $C$ after the inversion is $C$ itself, $C$ must be perpendicular to $M$.
\end{proof} 

A system of circles $A$, $B$, $C$, and $M$ meeting the conditions of the lemma is depicted in \cref{fig:antisimilitude}.

\subsection{The Main Lemma}

Recall \cref{lem:biclique}: In any $(3,b)$-tangency biclique in $\mathbb{E}^D$ or $\mathbb{H}^D$ defined by spheres $A_i$ and $B_j$, the centers of the spheres $B_j$ all lie on a common hyperplane.

\begin{proof}[Proof of \cref{lem:biclique}]
If any one of the spheres $A_i$ is contained in another sphere $A_j$, we may assume by renumbering the spheres if necessary that $A_1$ is contained in $A_2$. Then in order for the spheres $B_j$ to be tangent to $A_1$, without crossing into $A_2$, $A_1$ must be tangent to $A_2$ (externally to $A_1$ and internally to $A_2$), and all of the spheres $B_j$ must share this same point of tangency. The same reasoning applied to any two of the spheres $B_j$ shows that $A_3$, also, shares the same point of tangency. Thus in this special case, all of the sphere centers lie on a Euclidean or hyperbolic line, the line perpendicular to all of the spheres through this point of tangency.

Otherwise, if the $(3,b)$-tangency biclique is given in $\mathbb{H}^D$, then we view it as part of a Poincar\'e ball model of $\mathbb{H}^D$, giving us a corresponding $(3,b)$-tangency biclique in $\mathbb{E}^D$ (but with different points as the sphere centers). This allows us to apply \cref{lem:antisimilitude-orthogonal} to the spheres $A_1$, $A_2$, and their tangent spheres $B_j$, giving one sphere of antisimilitude $M_{12}$ orthogonal to all of the spheres $B_j$. By applying \cref{lem:antisimilitude-orthogonal} a second time to the spheres $A_1$, $A_3$, and their tangent spheres $B_j$, we obtain a second sphere of antisimilitude $M_{13}$ orthogonal to all of the spheres $B_j$. These two spheres $M_{12}$ and $M_{13}$ are distinct because inversion through them maps $A_1$ to the two distinct spheres $A_2$ and $A_3$. If the $(3,b)$-tangency biclique was given in $\mathbb{H}^D$ then they remain orthogonal to all of the spheres $B_j$ in $\mathbb{H}^D$, because the Poincar\'e model is conformal (it preserves angles).

By applying \cref{lem:2sphere-orthogonal} to $M_{12}$ and $M_{13}$ we obtain a hyperplane in $\mathbb{E}^D$ or $\mathbb{H}^D$ that contains the center points of all bounded spheres that are orthogonal to $M_{12}$ and $M_{13}$. In particular this hyperplane contains the center points of all of the given spheres $B_j$.
\end{proof}

\section{Conclusions}

We have extended the Erd\H{o}s--Anning theorem on integer distances to arbitrary-dimension hyperbolic spaces, and made it quantitative also for arbitrary-dimension Euclidean spaces. A natural question raised by our work and previous work on generalizations of the Erd\H{o}s--Anning theorem to non-Euclidean metrics~\cite{Epp-JoCG-26} is which other metrics obey analogous theorems. For instance, convex distance functions on $\mathbb{R}^D$ have badly-behaved Voronoi diagrams for which proofs based on $D$-tuples of landmarks may be problematic~\cite{IckKleLe-FI-95}; might it be possible to generalize the theorem to these methods using larger numbers of landmarks?

\bibliographystyle{plainurl}
\bibliography{tangent}

\end{document}